\documentclass [12pt,a4paper]{article}

\usepackage{pslatex}
\usepackage{amsfonts,color,morefloats}
\usepackage{amssymb,amsmath,latexsym,amsthm}
\usepackage{bbm}
\usepackage{amsfonts}
\usepackage{mathrsfs}
\usepackage{amsthm}
\usepackage{latexsym}
\usepackage{color}
\usepackage{dcolumn}
\usepackage{extarrows}
\usepackage{indentfirst}
\usepackage{makecell}
\newtheorem{theorem}{Theorem}[section]
\newtheorem{lemma}[theorem]{Lemma}

\newtheorem{corollary}[theorem]{Corollary}

\newtheorem{definition}{Definition}[section]
\newtheorem{example}{Example}[section]
\newtheorem{remark}{Remark}[section]

\begin{document}

\begin{center}
\textbf{\Large{Constructions on Real Approximate Mutually Unbiased Bases } }\footnote {The work was supported by
the National Natural Science Foundation of China (NSFC) under Grant 12031011 and 11701553.

Minghui Yang is the State Key Laboratory of Information Security, Institute of Information Engineering, Chinese Academy of Sciences, Beijing 100093,  China, (e-mail: yangminghui6688@163.com).

Aixian Zhang is with the Department of Mathematical Sciences, Xi'an University of Technology, 710054, China, (e-mail: zhangaixian1008@126.com).

Jiejing Wen is with School of Cyber Science and Technology, Shandong University,  Qingdao, 266237, China, (e-mail: jjwen@sdu.edu.cn).

Keqin Feng is with the Department of Mathematical Sciences, Tsinghua University, Beijing, 100084, China, (e-mail: fengkq@tsinghua.edu.cn).}
\end{center}

\begin{center}
\small Minghui Yang, Aixian Zhang, Jiejing Wen, and Keqin Feng
\end{center}

%-------------------------------------------------------------------------

\begin{abstract}
 Mutually unbiased bases (MUB) have many applications in quantum information processing and quantum cryptography. Several complex MUB's in $\mathbb{C}^d$ for some dimension $d$ and with larger size have been constructed. On the other hand, real MUB's with larger size are rare which lead to consider constructing approximate MUB (AMUB). In this paper we present a general and useful way to get real AMUB in $\mathbb{R}^{2d}$ from any complex AMUB in $\mathbb{C}^d$. From this method we present many new series of real AMUB's with parameters better than previous results.
\end{abstract}


\section{Introduction}

 Let $F$ be complex field $\mathbb{C}$ or real field $\mathbb{R}$, $d\geq 2$. For $v=(v_1,\ldots,v_d)$ and $u=(u_1,\ldots,u_d)\in F^d$, the inner product of $v$ and $u$ is $$(v, u)=\sum_{i=1}^dv_i\overline{u_i}\in F,$$
where $\overline{u_i}$ is the complex conjugate of $u_i$. This inner product is Hermitian for $F=\mathbb{C}$ and Euclidean for $F=\mathbb{R}$. If $(v,v)=1$, $v$ is called a unit vector. $(a_1,\ldots,a_d)$ is called an orthogonal basis in $F^d$ if $(a_i, a_j)=\delta_{ij}$ $(1\leq i,j\leq d).$ The usual example of such basis is
$$B_{*}=\{e_1,\ldots,e_d\}, e_1=(1,0,\ldots,0),  e_2=(0,1,\ldots,0),\ldots, e_d=(0,\ldots,0,1).$$

\begin{definition} \label{def1} Let $B_1,\ldots, B_n$ be orthogonal bases in $F^d$. $\mathscr{B}=\{B_1,\ldots, B_n\}$ is called mutually unbiased bases (MUB) if $$|(v, u)|=\frac{1}{\sqrt{d}}\  {for \ all}\  v\in B_\lambda, u\in B_\mu, 1\leq\lambda\neq\mu\leq n.$$
\end{definition}

MUB is one of the basic notions widely used in quantum information processing and quantum cryptography \cite{Be,Na,Pl,Woo}. One of the basic problems is to construct MUB in $F^d$ with larger size $n$. More precisely, for $d\geq2$, let
$$N_F(d)=\max \{n: {\rm{there}} \ {\rm{exists \ MUB }}\ \mathscr{B}=\{B_1,\ldots, B_n\} \ {\rm{in}} \ F^d \ {\rm{with}}\ {\rm{size}} \ n\}.$$

\noindent The basic problems in MUB theory are:

(1). To determine the value $N_F(d)$ or present upper and lower bounds.

(2). To construct MUB in  $F^d$ with larger size $n$.

In mathematical respect, both problems belong to a branch of combinatorics, called spherical $t$-design. Each MUB on $F^d$ is a special type of spherical $1$-design in the sphere $\{v\in F^d: (v,v)=1\}$. From this point of view we get the upper bound of $N_{\mathbb{C}}(d)\leq d+1$ and  $N_{\mathbb{R}}(d)\leq \frac{d}{2}+1$ (Welch bound). An MUB $\mathscr{B}$ reaches the Welch bound if and only if $\mathscr{B}$ is a spherical $2$-design. On the other hand, the lower bound of $N_F(d)$ can be derived by constructing MUB's in $F^d$ for several $d$ using tools from number theory (evaluation of character sums) and combinatorial theory (mutually orthogonal Latin squares, Hadamard matrices (\cite{Bo,Go,Kla,Wo}).

In the complex case $F=\mathbb{C}$, it is shown that for $d=p^m$, a power of prime, $N_\mathbb{C}(p^m)$ reaches the upper bound $p^m+1$. For all $d\neq p^m$,   $N_\mathbb{C}(p^m)\geq 3$, but to construct complex MUB in $\mathbb{C}^d$ for $d\equiv 2\pmod 4$ with larger size is quite difficult. In the real case $F=\mathbb{R}$, $N_\mathbb{R}(d)=1$ for $4\nmid d$; $N_\mathbb{R}(d)\leq 2$ for $d=4m$ and $n$ is not a square; and $N_\mathbb{R}(d)\leq 3$ for $d=4m^2$ and $2\nmid m$. Namely, there does not exist real MUB with size $4$ in $\mathbb{R}^d$ if $d\neq(4m)^2$. These results lead several works on constructing approximate MUB (AMUB, see Section \ref{sec3} for its definition). Several series of AMUB's has been presented (\cite{Sh,Li,Wan} for $F=\mathbb{C}, \cite{Ku}$ for $F=\mathbb{R}$). In this paper we focus on the real AMUB. We present a general method to get real AMUB in $\mathbb{R}^{2d}$ from any complex AMUB in $\mathbb{C}^d$. From this result we can obtain many series of real AMUB's with new parameters and such parameters are better than ones of previous results.

In Section \ref{sec2} we introduce several combinatorial subjects as spherical $t$-design, mutually orthogonal Latin squares (MOLS), Hadamard matrix, and their relationship with MUB's. With these relationship we introduce and demonstrate the bounds of $N_F(d)$ and constructing methods of known complex and real MUB's. Then we introduce known constructions on AMUB's in Section \ref{sec3}. Our new result of construction on real AMUB's is presented in Section \ref{sec4}. The last section is conclusion.

\section{Preliminaries}\label{sec2}

In this section we introduce several combinatorial objects: spherical $t$-design, mutually orthogonal Latin squares (MOLS) and Hadamard matrix. We show the relationship between these combinatorial objects with MUB's, introduce and demonstrate the bounds of $N_F(d)$ and constructions of MUB's by such relationship.

\subsection{Spherical $t$-design, MOLS and Hadamard matrix}

\subsubsection{Spherical $t$-design}

$F$ denotes the complex field $\mathbb{C}$ or the real field $\mathbb{R}$.

\begin{lemma}(Welch bound \cite{Wa}) \label{welch bound} Let $C$ be a finite set of unit vectors in $F^d (d\geq 2), |C|=n>d$. Then for each integer $t\geq1$,

\begin{equation*} \frac{1}{n^2}\sum_{u,v\in C}|(u,v)|^{2t}\geq
\begin{cases} 1/ \binom {d+t-1} t,& \textrm{for $F=\mathbb{C}$} \\\frac{(2t-1)!!}{d(d+2)\cdots(d+2(t-1))}, & \textrm{for $F=\mathbb{R}$  }
\end{cases}
\end{equation*}
where $(2t-1)!!=1\cdot3\cdot5\cdots(2t-1)$.
\end{lemma}

\begin{definition}
If the equality in Lemma \ref{welch bound} holds for some $t\geq 1$, $C$ is called a spherical $t$-design in $F^d$.
\end{definition}

Now we show that each MUB forms a spherical $1$-design. Let $\mathscr{B}=\{B_1,\ldots,B_n\}$ be $n$ orthonormal bases in $F^d$. We define

$$\gamma_F(\mathscr{B})={\rm{max}}\{|(v_i, v_j)|: v_i\in B_i, v_j\in B_j, 1\leq i\neq j\leq n\}.$$

\noindent There are $nd$ vectors in $\mathscr{B}$. For $v, u\in\mathscr{B}$,
\begin{align*}
&|(v,u)| \\= &\begin{cases} 1, & \textrm{if $v=u$} \\0, & \textrm{if $v\neq u$ and $v, u$ belong to the same $B_\lambda$ for some $\lambda~ (1\leq\lambda\leq n)$}\\ \leq \gamma_F(\mathscr{B}), & \textrm{otherwise.}
\end{cases}
\end{align*}

\noindent From  Lemma \ref{welch bound} we get
\begin{equation*} \frac{1}{n^2d^2}(nd+n(n-1)d^2{\gamma_F(\mathscr{B})}^{2t})\geq
\begin{cases} 1/ \binom {d+t-1} t,& \textrm{for $F=\mathbb{C}$} \\\frac{(2t-1)!!}{d(d+2)\cdots(d+2(t-1))} & \textrm{for $F=\mathbb{R}$  }
\end{cases}
\end{equation*}
\noindent Then we get a lower bound of $\gamma_F(\mathscr{B})$.

\begin{theorem}\label{th2.2}
Let $\mathscr{B}$ be $n$ orthonormal bases in $F^d$. Then for $t\geq 1$,
$${\gamma_F(\mathscr{B})}^{2t}\geq W_F(d,n;t)\ \ (\textrm{Welch bound of $\mathscr{B}$})$$
\noindent where
\begin{equation*} W_F(d,n;t)=
\begin{cases} \frac{n}{(n-1) \binom {d+t-1} t}-\frac{1}{d(n-1)},& \textrm{for $F=\mathbb{C}$} \\\frac{n(2t-1)!!}{(n-1)d(d+2)\cdots(d+2(t-1))}-\frac{1}{d(n-1)}, & \textrm{for $F=\mathbb{R}$.  }
\end{cases}
\end{equation*}
\noindent Particularly, for $t=1$, ${\gamma_F(\mathscr{B})}^2\geq\frac{1}{d}$ and for $t=2$,
\begin{equation*} {\gamma_F(\mathscr{B})}^4\geq
\begin{cases}\frac{2n-(d+1)}{d(d+1)(n-1)},& \textrm{for $F=\mathbb{C}$} \\\frac{3n-(d+2)}{(n-1)d(d+2)}, & \textrm{for $F=\mathbb{R}$.  }
\end{cases}
\end{equation*}
Moreover, $\mathscr{B}$ is an MUB in $F^d$ if and only if $\gamma_F(\mathscr{B})=\frac{1}{\sqrt{d}}$ (so that $|(v_i,v_j)|=\frac{1}{\sqrt{d}}$ for all pairs $v_i$ and $v_j$ belong to different orthonormal bases of $\mathscr{B}$), which also equivalent to that the $nd$ unit vectors in $\mathscr{B}$ form a spherical $1$-design in $F^d$.
\end{theorem}
\subsubsection{MOLS}
\begin{definition}
For $d\geq 2$, a Latin square of order $d$ is a $d\times d$ matrix such that each row and each column is a permutation of $\{1,2,\ldots,d\}$. Latin squares $L=(a_{ij})$ and $L'=(a'_{ij})$ of order $d$ are called orthogonal if
$$\{(a_{ij}, a'_{ij}): (1\leq i,j\leq d)\}=\{(\lambda, \mu): 1\leq\lambda,\mu\leq d\}$$
\end{definition}

\begin{example}
The following two Latin squares of order $3$ are orthogonal.\\
$$L_1=\left(
  \begin{array}{ccc}
    1 & 2 & 3 \\
    2 & 3 & 1 \\
    3 & 1 & 2 \\
  \end{array}
\right),L_2=\left(
  \begin{array}{ccc}
    1 & 2 & 3 \\
    3 & 1 & 2 \\
    2 & 3 & 1 \\
  \end{array}
\right),$$ $$\left(
  \begin{array}{ccc}
    (1,1) & (2,2) & (3,3) \\
    (2,3) & (3,1) & (1,2) \\
    (3,2) & (1,3) & (2,1) \\
  \end{array}
\right)~(orthogonality)$$
\end{example}
Let $M(d)$ be the maximal size of mutually orthogonal Latin squares of order $d$,
$$M(d)={\rm{max}} \{n: {\rm{there}}\ {\rm{exist}} \ n \ {\rm{mutually}}\ {\rm{orthogonal}} \ {\rm{Latin}}\ {\rm{squares}}\ {(\rm{MOLS})}\ {\rm{of}}\ {\rm{order}}\ d\}. $$
The following results are well-known.
\begin{lemma}\label{lem2.1}
For $d\geq 2$

(1). $M(d)\leq d-1$, and for prime-power $d=p^e, M(p^e)=p^e-1$.

(2). For $d_1,$ $d_2\geq 2$, $M(d_1d_2)$$\geq$ ${\rm{min}}\{M(d_1), M(d_2)\}$. Therefore for $d=p_1^{e_1}\cdots p_g^{e_g}$\\$\geq 2$,
$$M(d)\geq {\rm{min}}\{M(p_1^{e_1}),\ldots,M(p_g^{e_g})\}={\rm{min}}\{p_1^{e_1}-1,\ldots,p_g^{e_g}-1 \}.$$
Particularly, if $d\not\equiv2\pmod 4$ and $d\geq 4$, then $M(d)\geq 3$.

(3). $M(2)=M(6)=1$ and for all other $d\equiv 2\pmod 4$, $M(d)\geq 2.$
\end{lemma}
\subsubsection{Hadamard matrix}
\begin{definition}
For $d\geq2$, a $d\times d$ matrix $H=(h_{ij})_{1\leq i,j\leq d}$ is called Hadamard matrix if $h_{ij}\in\{\pm 1\} (1\leq i,j\leq d)$ and $HH^T=dI_d$ where
$I_d$ is the identity matrix of order $d$.
\end{definition}
For $d=2$, $H_2=\left(
                  \begin{array}{cc}
                    1 & 1 \\
                    1 & -1 \\
                  \end{array}
                \right)
$ is a Hadamard matrix. It is easy to show that if there exists a Hadamard matrix of order $d\geq 3$, then $d\equiv0\pmod 4$. One of famous conjectures in combinatorial theory is that the Hadamard matrix of order $d$ exists for all positive integer $d\equiv0\pmod 4$. The conjecture has been verified for all such $d$ up to a very large integer.
\subsubsection{Known results on $N_F(d)$ and constructions of MUB's}

With combinatorial (as above introduced) and number-theoretical tools (evaluation of character sums), we can introduce and demonstrate the following known results on bounds of $N_F(d)$ and constructions on MUB's.

(A). Complex case $(F=\mathbb{C})$
\begin{theorem}\label{bounds}
For $d\geq 2$,

(1). $N_\mathbb{C}(d)\leq d+1$. Moreover, for $d+1$ orthonormal bases $\mathscr{B}=\{B_1,\ldots,B_{d+1}\}$ in $\mathbb{C}^d$, $\mathscr{B}$ is an MUB (so that $N_\mathbb{C}(d)=d+1)$ if and only if the $d(d+1)$ vectors in $\mathscr{B}$ form a spherical $2$-design in $\mathbb{C}^d$.

(2). (\cite{Kla}) For a prime-power $d=p^m, N_\mathbb{C}(d)=d+1$.

(3). (\cite{Kla}) For any $d_1,d_2\geq 2$, $N_\mathbb{C}(d_1d_2)\geq min \{N_\mathbb{C}(d_1), N_\mathbb{C}(d_2)\}$. As a direct consequence, let $d=p_1^{e_1}\cdots p_g^{e_g}\geq 2$ be the decomposition of $d$ into product of prime powers. Then
$$N_\mathbb{C}(d)\geq \min \{N_\mathbb{C}(p_1^{e_1}),\ldots,N_\mathbb{C}(p_g^{e_g})\}=\min\{p_1^{e_1}+1,\ldots, p_g^{e_g}+1\}\geq 3.$$

(4). (\cite{Wo}) $N_\mathbb{C}(d^2)\geq M(d)+2$ where $M(d)$ is the maximal size of MOLS of order $d$. Particularly, $M_\mathbb{C}(d^2)\geq 5$ for all $d \not\equiv2\pmod 4$ and $d\geq 6$.
\end{theorem}
\begin{proof} (1) Let $\mathscr{B}=\{B_1,\ldots,B_n\}$ be an MUB in $\mathbb{C}^d$. Then $\gamma_\mathbb{C}(\mathscr{B})=\frac{1}{\sqrt{d}}$. Take $t=2$ in
Theorem \ref{th2.2}, we get $\frac{1}{d^2}={\gamma_\mathbb{C}(\mathscr{B})}^4\geq\frac{2n-(d+1)}{d(d+1)(n-1)}$ which implies that $n\leq d+1$, and $n=d+1$ if and only if ${\gamma_\mathbb{C}(\mathscr{B})}^4=\frac{2n-(d+1)}{d(d+1)(n-1)}$, namely, the $d(d+1)$ vectors in $\mathscr{B}$
form a spherical $2$-design in $\mathbb{C}_d$.

(2). There exist several constructions on MUB's in $\mathbb{C}^{p^m}$ with size $p^m+1$. One of them is given in \cite{Kla} by using quadratic Gauss sums over finite field $\mathbb{F}_{p^m}$ for $p\geq 3$ and the Galois ring $GR(4,m)$ for $p=2$.

(3) is proved by tensor construction. For the proof of (4) we refer to \cite{Wo}.

(4) is by Lemma \ref{lem2.1}(3).
\end{proof}
(B). Real case $(F=\mathbb{R})$

Firstly we present a characterization of real MUB in terms of Hadamard matrices. For a basis $B=\{v_1,\ldots,v_d\}$ of $\mathbb{R}^d$, we get a $d\times d$ real matrix
$$M(B)=\left(
         \begin{array}{c}
           v_1 \\
          \vdots\\
           v_d \\
         \end{array}
       \right).
$$
It is easy to see that

(\uppercase\expandafter{\romannumeral1}). $B$ is an orthonormal basis of $\mathbb{R}^d$ if and only if $M(B)$ is an orthogonal matrix, namely, $M(B){M(B)}^T=I_d$.

Let $O_d(\mathbb{R})$ be the orthogonal group of order $d$. Each orthogonal matrix $A\in O_d(\mathbb{R})$ acts on space $\mathbb{R}^d$ and keeps the inner product in $\mathbb{R}^d$: for $v, u\in \mathbb{R}^d$, $(vA, uA)=(v, u)$. Then from the definition of MUB and (\uppercase\expandafter{\romannumeral1}) we know that $\{B_0,\ldots,B_{n-1}\}$ is an MUB in $\mathbb{R}^d$ if and only if $\{B_0A,\ldots, B_{n-1}A\}$ is an MUB in $\mathbb{R}^d$ where for $B_\lambda=\{v_1,\ldots,v_d\}$, $B_\lambda A=\{v_1 A,\ldots,v_d A\}$.

Taking $A=M{(B_0)}^{-1}$, then $M(B_0)A=I_d$ which means that $B_0A=\{e_1,\ldots,e_d\}$\\$=B_{\ast}$, where $e_1=(1,0,\ldots,0)$, $e_2=(0,1,\ldots,0)$, $\ldots$,
$e_d=(0,\ldots,0,1)$. Without loss of generality we can assume $B_0=B_{\ast}$ and $\mathscr{B}=\{B_0,B_1,\ldots,B_{n-1}\}$ is $n$ orthonormal bases in $\mathbb{R}^d$.

Let $H_\lambda=\sqrt{d}M(B_\lambda) (0\leq\lambda\leq n-1)$. Then $H_0=\sqrt{d}I_d$ and for $1\leq\lambda\leq n-1$,

(\uppercase\expandafter{\romannumeral2}). $B_0=B_{\ast}$ and $B_\lambda$ are unbiased, namely,
$$\frac{1}{\sqrt{d}}=|(e_i, v)|=|v_i|\ (1\leq i\leq d)\ {\rm{for\  each}}\  v=(v_1,\ldots,v_d)\in B_\lambda$$
if and only if all entries of $H_\lambda$ are $1$ or $-1$ which means that all $H_\lambda (1\leq\lambda\leq n-1)$ are Hadamard matrices since $H_\lambda H_\lambda^T=dM(B_\lambda){M(B_\lambda)}^T=dI_d$.

Moreover, if $n\geq 3$ then for $1\leq i\neq j\leq n-1$ we have

(\uppercase\expandafter{\romannumeral3}) $B_i$ and $B_j$ are unbiased, namely, all entries of $M(B_i){M(B_j)}^T=\frac{1}{d}H_iH_j^T$ are $\frac{1}{\sqrt{d}}$
or $-\frac{1}{\sqrt{d}}$, if and only if $H_{ij}=\frac{1}{\sqrt{d}}H_iH_j^T$ are Hadamard matrix since
$$H_{ij}H_{ij}^T=\frac{1}{d}(H_{i}H_{j}^T){(H_{i}H_{j}^T)}^T=\frac{1}{d}H_i(H_{j}^TH_j)H_{i}^T=dI_d.$$

In summary, we get the following criterion on real MUB in terms of Hadamard matrices which essentially given in \cite{Bo}.

\begin{lemma}\label{ha}
For $d\geq 2$, there exists a real MUB in $\mathbb{R}^d$ with size $n\geq 2$ if and only if there exist Hadamard matrices $H_1,\ldots, H_{n-1}$ of order $d$,
and if $n\geq 3$, all $H_{ij}=\frac{1}{\sqrt{d}}H_iH_j^T (1\leq i< j\leq n-1)$ are also Hadamard matrices.
\end{lemma}

\begin{remark}\label{le}
From Lemma \ref{ha} we know that $N_\mathbb{R}(d)\geq 2$ if and only if there exists a Hadamard matrix of order $d=4m$.
\end{remark}

Now we introduce the known results on $N_\mathbb{R}(d)$. It is obvious that $N_\mathbb{R}(d)\leq N_\mathbb{C}(d)$. The following results show that $N_\mathbb{R}(d)$ is at most $3$ if $d\neq {(4m)}^2$.

\begin{theorem}\label{2.4}(\cite{Bo,Ca})

(1). $N_\mathbb{R}(d)\leq\frac{d}{2}+1$ for all $d\geq 3$. Moreover, for $\frac{d}{2}+1$ orthonormal bases $\mathscr{B}$ in $\mathbb{R}^d$, $\mathscr{B}$ is a real MUB (so that $N_\mathbb{R}(d)=\frac{d}{2}+1)$ if and only if the $d(\frac{d}{2}+1)$ vectors of $\mathscr{B}$ form a spherical $2$-design in $\mathbb{R}^d$.

(2). $N_\mathbb{R}(2)=2$. If $4\nmid d\geq 3$, then $N_\mathbb{R}(d)=1$. Namely, there is no pair of unbiased orthonormal basis in $\mathbb{R}^d$ for $4\nmid d, d\geq 3$.

(3). If $d=4m$ and $m$ is not a square, then $N_\mathbb{R}(d)\leq 2$ and $N_\mathbb{R}(d)=2$ if and only if there exists a Hadamard matrix of order $d=4m$.

(4). If $d=4m^2$ and $2\nmid m$, then $N_\mathbb{R}(d)\leq 3$. Moreover, if there exists a Hadamard matrix of order $d$, then $N_\mathbb{R}(d)\geq 2$.

(5). (\cite{Wo}) If $d=(4m)^2$, then $N_\mathbb{R}(d)\geq M(4m)+2$ where $M(l)$ is the maximal size of MOLS of order $l$. Particularly, for all $d={(4m)}^2$, $N_\mathbb{R}(d)\geq 4$.
\end{theorem}
\begin{proof}
(1) $N_\mathbb{R}(d)\leq\frac{d}{2}+1$ has been proved in \cite{Ca} and \cite{Bo}. In fact, this is a direct consequence of Welch bound for $t=2$. Let $\mathscr{B}=\{B_1,\ldots,B_n\}$ be an MUB in $\mathbb{R}^d$. Then $\gamma_\mathbb{R}(\mathscr{B})=\frac{1}{\sqrt{d}}$ and by taking $t=2$ in Theorem \ref{th2.2}, we get $\frac{1}{d^2}={\gamma_\mathbb{R}(\mathscr{B})}^4\geq \frac{3n-(d+2)}{(n-1)d(d+2)}$ which implies that $n\leq \frac{d}{2}+1$, and
$n=\frac{d}{2}+1$ if and only if the $d(\frac{d}{2}+1)$ vectors in $\mathscr{B}$ form a spherical $2$-design in $\mathbb{R}^d$.

(2) From (1) we know that $N_\mathbb{R}(2)\leq\frac{2}{2}+1=2.$ Then by Remark \ref{le} we know $N_\mathbb{R}(2)=2$. In fact, $B_{\ast}=\{e_1=(1,0), e_2=(0,1)\}$ and $B=\{\frac{1}{\sqrt{2}}(1,1), \frac{1}{\sqrt{2}}(1,-1)\}$ form an MUB in $\mathbb{R}^2$.

Assume $d\geq 3$ and $N_\mathbb{R}(d)\geq 2$. By Lemma \ref{ha}, there exists a Hadamard matrix of order $d$ which implies $4|d$.

(3) If $d=4m$ and $N_\mathbb{R}(d)\geq 3$, there exist Hadamard matrices $H_1$ and $H_2$ of order $d$ such that $\frac{1}{\sqrt{d}}H_1H_2^T$ is also Hadamard matrix which implies $\sqrt{d}\in \mathbb{Z}$ so that $m$ is a square.

(4) If $d=4m^2$ and $N_\mathbb{R}(d)\geq 4$, there exist Hadamard matrices $H_1$, $H_2$, $H_3$ of order $d$ and $H_{ij}=\frac{1}{2^m}H_iH_j^T$ are
also Hadamard matrices for all $1\leq i\neq j\leq 3$. Consider the equality
$$(H_1+H_2)(H_1+H_3)^T=dI_d+2m(H_{13}+H_{21}+H_{23}).$$
\noindent All entries of the matrix in the left-hand side are divided by $4$. In the right-hand side, all entries $\equiv d+2m\equiv 2m\pmod 4$.
This implies $2|m$.

For the proof of (5) we refer to \cite{Wo}.
\end{proof}

\section{Approximate Mutually Unbiased Bases (AMUB)}\label{sec3}

To construct complex MUB with larger size in $\mathbb{C}^d$ for non-power of prime $d$ is quite difficult. For the real case, Theorem \ref{2.4} shows that there is no real MUB with size $4$ in $\mathbb{R}^d$ if $d\neq{(4m)}^2$. Such situation leads several works to consider approximate MUB by loosing requirement $\gamma_F(\mathscr{B})=O(\frac{1}{\sqrt{d}})$ instead of $\gamma_F(\mathscr{B})=\frac{1}{\sqrt{d}}$. More precisely,
\begin{definition}
Let $\{d_\lambda\}_{\lambda=1}^{\infty}$ and $\{n_\lambda\}_{\lambda=1}^{\infty}$ be two series of positive integers, $d_\lambda\rightarrow\infty$
when $\lambda\rightarrow\infty$. Let $\mathscr{B}_\lambda$ be a set of $n_\lambda$ orthonormal bases in $F^{d_\lambda} (\lambda=1, 2, \ldots)$. The series $\{\mathscr{B}_\lambda\}_{\lambda=1}^{\infty}$ is called $\frac{c}{\sqrt{d_\lambda}}$-approximate MUB's ($\frac{c}{\sqrt{d_\lambda}}$-AMUB's) if $$
\mathop{\overline{lim}}\limits_{\lambda\rightarrow\infty}\gamma_F(\mathscr{B}_\lambda)\big/\frac{c}{\sqrt{d_\lambda}}=1,$$where $c$ is a real constant independent of $d_\lambda$ and $n_\lambda$.
\end{definition}
Several series of AMUB's have been presented (\cite{Sh,Wan,Ku,Li}). In this section we introduce main previous constructions on AMUB's.
\subsection{Complex AMUB's}

Firstly we introduce the construction in \cite{Sh} by using character sums on elliptic curves over finite fields. Let $p\geq 5$ be a prime. An elliptic curve over $\mathbb{F}_p$ is defined by
$$E: y^2=g(x), g(x)=x^3+ax+b\in \mathbb{F}_p[x], 4a^3+27b^2\neq 0.$$
The set of $\mathbb{F}_p$-points on $E$ (including the infinite point $\infty$)
$$E(\mathbb{F}_p)=\{D=(A,B)\in \mathbb{F}_p^2: B^2=g(A)\}\cup \{\infty\}$$
forms a finite abelian group with zero element $\infty$. By Hasse-Weil theorem, the order $d=|E(\mathbb{F}_p)|$ satisfies
\begin{equation}\label{1}
p-2\sqrt{p}\leq d\leq p+2\sqrt{p}~ ({\rm{namely,}} \sqrt{d}-1\leq\sqrt{p}\leq\sqrt{d}+1)
\end{equation}
Moreover, it is shown by Deuring (1941) that for any integer $d$ satisfying (\ref{1}), there exists an elliptic curve $E$ over $\mathbb{F}_p$ such that $|E(\mathbb{F}_p)|=d$.

The function field of $E$ is $K=\mathbb{F}_p(x,y)$ where $y=\sqrt{g(x)}=\sqrt{x^3+ax+b}$. The ring of integral elements in $K$ is
$$O_K=\mathbb{F}_p[x,y]=\{f(x)=u(x)+v(x)y: u(x), v(x)\in \mathbb{F}_p[x]\}.$$
Let ${\rm{Deg}}(x)=2$ and ${\rm{Deg}}(y)=3$. Then for $f(x)=u(x)+v(x)y\in O_K$,
\begin{equation} \label{2}
{\rm{Deg}}(f)={\rm{max}}\{{\rm{Deg}}(u(x)), 3+{\rm{Deg}}(v(x)) \}={\rm{max}}\{2\deg\ u(x),3+2\deg\ v(x) \}
\end{equation}
where $\deg u(x)$ is the usual degree of polynomial $u(x)\in \mathbb{F}_p[x].$ In fact, ${\rm{Deg}}(f)=-v_\infty(f)$ where $v_\infty$ is the $\infty$-adic exponential valuation of $K$.

For $2\leq m\leq d-1$, consider the following subset of $O_K$
$$\sum_m=\{f(x,y)=u(x)+v(x)y\in O_k: {\rm{Deg}}(f)\leq m, u(0)=0\}.$$
Then $\sum_m$ is an $\mathbb{F}_p$-subspace of $O_K$ and by formula (\ref{2}) $|\sum_m|=p^{m-1}$.

Let  $\widehat{E(\mathbb{F}_p)}$ be the character group of the finite abelian group $E(\mathbb{F}_p)$. From the $L$-function theory on elliptic curves we have the following estimation on character sums.
\begin{lemma}\label{3.1}(\cite{Sh}) For $f\in O_K\backslash \{constant \ functions\}$, $\chi\in \widehat{E(\mathbb{F}_p)}$ and the character sum
$$S_E(f,\chi)=\sum_{P\in E(\mathbb{F}_p)}\zeta_p^{f(P)}\chi(P)~(\zeta_p=e^{\frac{2\pi\sqrt{-1}}{p}}),$$
we have $|S_E(f,\chi)|\leq2\cdot Deg(f)\sqrt{p}$ where $f(\infty)=0$ and for $P=(A, B)\in E(\mathbb{F}_p),$ $f(P)=f(A, B)$.
\end{lemma}

Now for $f\in O_K$ and $\chi\in\widehat{E(\mathbb{F}_p)}$ we define a unit vector in $\mathbb{C}^d$, $d=|E(\mathbb{F}_p)|$,
$$v_{f,\chi}=\frac{1}{\sqrt{d}}(\zeta_p^{f(P)}\chi(P))_{P\in E(\mathbb{F}_p)}.$$
For each $f\in O_K$, $B_f=\{v_{f,\chi}: \chi\in\widehat{E(\mathbb{F}_p)}\}$ is an orthonormal basis of $\mathbb{C}^d$ by the orthogonal relation of characters in $\widehat{E(\mathbb{F}_p)}$. By Lemma \ref{3.1} we can see that for $f, g\in \sum_m, f\neq g, v_{f,\chi}\in {B_f}, v_{g,\chi'}\in B_g$,
$$|\langle v_{f,\chi},v_{g,\chi'}\rangle|=\big|\frac{1}{d}\sum_{P\in E(\mathbb{F}_p)}\zeta_p^{f(P)-g(P)}\chi\overline{\chi'}(P)\big|\leq\frac{1}{d}\cdot2m\sqrt{p}\leq \frac{2m(\sqrt{d}+1)}{d}.$$
We get the following real AMUB.
\begin{theorem}\label{th3.2}
(\cite{Sh}) Let $p\geq 5$ be a prime, $p-2\sqrt{p}\leq d\leq p+2\sqrt{p}$. Then for each $m, 2\leq m\leq d-1, \mathscr{B}=\{B_f: f\in \sum_m\}$ is
an AMUB in $\mathbb{R}^d$ with $n=p^{m-1}=|\sum_m|$ orthonormal bases $B_f$ and $\gamma_\mathbb{C}(\mathscr{B})\leq \frac{2m}{\sqrt{d}}+\frac{2m}{d}$.
\end{theorem}

Next we introduce the constructions of real AMUB's in \cite{Wan,Li} by using Gauss sums and Jacobi sums over finite fields and Galois rings.
\begin{theorem}\label{th3.3}
Let $q$ be a power of prime, $\widehat{\mathbb{F}}_q$ and $\widehat{\mathbb{F}}_q^{\ast}$ be the groups of additive and multiplicative characters of finite field $\mathbb{F}_q$ respectively.

(1). (\cite{Wan}) Let $d=q-1, \mathbb{F}_q\backslash \{0,1\}=\{x_1,\ldots,x_{q-2}\}$. For $\chi,\chi'\in \widehat{\mathbb{F}}_q^{\ast}$, let
$$v_{\chi,\chi'}=\frac{1}{\sqrt{d}}(\chi(x_1)\chi'(1-x_1),\ldots,\chi(x_{q-2})\chi'(1-x_{q-2}),1)\in\mathbb{C}^d$$
$$B_\chi=\{v_{\chi,\chi'}: \chi'\in \widehat{\mathbb{F}}_q^{\ast}\}$$
$$\mathscr{B}=\{B_\chi: \chi\in \widehat{\mathbb{F}}_q^{\ast}\}\cup \{B_{\ast}\}, \ n=|\mathscr{B}|=q=d+1.$$
Then, by using Jacobi sums, $\mathscr{B}$ is an AMUB in $\mathbb{C}^d$ with $\gamma_\mathbb{C}({\mathscr{B}})\leq (\frac{1}{d}+\frac{2\sqrt{d+1}+1}{d^2})^{\frac{1}{2}}$.

(2).(\cite{Wan}) Let $d=q-1$, and $$v_{\lambda,\chi}=\frac{1}{\sqrt{d}}(\lambda(x)\chi(x))_{x\in\mathbb{F}_q^{\ast}}\in \mathbb{C}^d~(\lambda\in \widehat{\mathbb{F}}_q, \lambda\in \widehat{\mathbb{F}}_q^{\ast})$$
$$B_\lambda=\{v_{\lambda,\chi}: \chi\in \widehat{\mathbb{F}}_q^{\ast}\}$$
$$\mathscr{B}=B_{\ast}\cup \{B_\lambda: \lambda\in \widehat{\mathbb{F}}_q\}. \ n=|\mathscr{B}|=q+1=d+2.$$
Then, by Gauss sums, $\mathscr{B}$ is an AMUB in $\mathbb{C}^d$ with $\gamma_{\mathbb{C}}(\mathscr{B})\leq (\frac{1}{d}+\frac{1}{d^2})^{\frac{1}{2}}$.

(3). (\cite{Wan}) Using $\widehat{\mathbb{F}}_{q^2}^{\ast}$, it can be constructed an AMUB $\mathscr{B}$ in $\mathbb{C}^d (d=q+1)$ with $n=|\mathscr{B}|=d-1$
and $\gamma_\mathbb{C}(\mathscr{B})\leq (\frac{1}{d}+\frac{2\sqrt{d-1}}{d^2})^{1/2}$.

(4). (\cite{Li}) Using the Gauss sums on Galois ring $GR(4,r)$, it can be constructed an AMUB $\mathscr{B}$ in $\mathbb{C}^d (d=q(q-1), q=2^r)$, with $n=|\mathscr{B}|=q+1$ and $\gamma_\mathbb{C}(\mathscr{B})\leq\frac{1}{q-1}(\sim\frac{1}{\sqrt{d}}).$
\end{theorem}

For real AMUB's, we introduce the following result.
\begin{theorem}(\cite{Ku})\label{3.4}
Let $d=(4m)^2$ and assume that there exists a Hadamard matrix of order $4m$. Then for each positive integer $n\leq 4m-3$, there exists a real AMUB
$\mathscr{B}$ in $\mathbb{R}^d$ with size $|\mathscr{B}|=n$ and $\gamma_\mathbb{R}(\mathscr{B})\leq \frac{\Delta(m,n)}{\sqrt{d}}$ where
$$\Delta(m,n)=\max\limits_{1\leq l\leq n-2}\gcd(4m, l).$$
\end{theorem}
\begin{corollary}\label{cor3.5}
Let $d={(4p)}^2$ and $p$ be a prime. Assume that there exists a Hadamard matrix of order $4p$. Then there exists a real AMUB $\mathscr{B}$ in $\mathbb{R}^d$ with $|\mathscr{B}|=p+1=\frac{\sqrt{d}}{4}+1$ and $\gamma_{\mathbb{R}}(\mathscr{B})\leq\frac{\sqrt{d}}{4}(=\frac{1}{p})$.
\end{corollary}
\begin{proof}
Take $m=p$ and $n=p+1$ in Theorem \ref{3.4}, we get a real AMUB $\mathscr{B}$ in $\mathbb{R}^d$ with $|\mathscr{B}|=n$ and $\gamma_{\mathbb{R}}(\mathscr{B})\leq\frac{\Delta(4p,p+1)}{\sqrt{d}}$ where
$$\Delta(4p, p+1)=\max_{1\leq l\leq p-1}\gcd(4p, l)\leq 4.$$
\end{proof}
\section{New Construction on Real AMUB's}\label{sec4}
By author's knowledge, Theorem \ref{3.4} is the only known construction on real AMUB's in $\mathbb{R}^d$ with larger size and restrictive $d=(4m)^2$. In this section we present a general construction of real AMUB $\mathscr{B}(\mathbb{R})$ in $\mathbb{R}^{2d}$ from any complex AMUB $\mathscr{B}$ in $\mathbb{C}^d$ with the same size $|\mathscr{B}(\mathbb{R})|=|\mathscr{B}|$ and $\gamma_\mathbb{R}(\mathscr{B}(\mathbb{R}))\leq \gamma_{\mathbb{C}}(\mathscr{B})$. From this construction and the known results on complex  AMUB's (including complex MUB's) we get many series of real AMUB's in $\mathbb{R}^D$ for large amount of even dimension $D$.

For $d\geq 2$ and $\gamma>0$, $F$ be $\mathbb{C}$ or $\mathbb{R}$, let
$$N_F(d;\gamma)=\max \{n:{\rm{there\ exists}}\ n\ {\rm{orthonormal \ bases }}  \mathscr{B}\  {\rm{in}} \ F^d \ {\rm{such \ that}}\  \gamma_F(\mathscr{B})\leq \gamma\}.$$

\noindent For $\gamma=\frac{1}{\sqrt{d}}, \gamma_F(\mathscr{B})\leq\frac{1}{\sqrt{d}}$ means that $\mathscr{B}$ is an MUB. We know that $N_F(d;\frac{1}{\sqrt{d}})=N_F(d)$.
\begin{theorem}\label{thm4.1}
For $d\geq 2$ and $\gamma>0$, $N_\mathbb{R}(2d; \gamma)\geq N_\mathbb{C}(d; \gamma)$. Namely, if there exists a complex $\frac{c}{\sqrt{d}}$-AMUB in $\mathbb{C}^d$ with size $n$, then there exists a real $\frac{c\sqrt{2}}{\sqrt{D}}$-AMUB in $\mathbb{R}^d (D=2d)$ with the same size $n$.
\end{theorem}
\begin{proof}Before proving this theorem, we consider a mapping from any complex vector $v=(v_1,\ldots, v_d)\in \mathbb{C}^d$ to two real vectors $v^{(1)}$ and $v^{(2)}$
in $\mathbb{R}^{2d}$ where
$$v^{(1)}=(R(v_1),I(v_1),\ldots,R(v_d),I(v_d))$$
$$v^{(2)}=(-I(v_1),R(v_1),\ldots,-I(v_d),R(v_d))$$
and $R(v_i)$, $I(v_i)$ denote the real and imaginary parts of $v_i\in\mathbb{C}$ respectively. The mapping from $v$ to $v^{(1)}$ and $v^{(2)}$ has the following properties.
\begin{lemma}\label{4.2}
(1). If $v$ is a unit vector, then $v^{(1)}$ and $v^{(2)}$ are unit vectors.

(2). For $v, u\in \mathbb{C}^{(d)}$ and $1\leq i, j\leq 2$, $|(v^{(i)}, u^{(j)}) |\leq |(v,u)|$. Particularly, if $(v,u)=0$, then $v^{(1)}, v^{(2)}, u^{(1)}$ and $u^{(2)}$ are orthogonal to each others.
\end{lemma}
\begin{proof}
(1) is proved from
$$(v,v)=\sum_{\lambda=1}^dv_{\lambda}\overline{v}_{\lambda}=\sum_{\lambda=1}^d({R(v_\lambda)}^2+{I(v_\lambda)}^2)=(v^{(1)}, v^{(1)})=(v^{(2)}, v^{(2)}).$$
(2). From
\begin{align*}
(v,u)& =\sum_{\lambda=1}^{d}v_{\lambda}\overline{u_\lambda}=\sum_{\lambda=1}^{d}(R(v_{\lambda})+iI(v_{\lambda}))(R(u_{\lambda})-iI(u_{\lambda}))~(i=\sqrt{-1})\\
&=\sum_{\lambda=1}^{d}[R(v_\lambda)R(u_\lambda)+I(v_\lambda)I(u_\lambda)]+i\sum_{\lambda=1}^{d}[I(v_\lambda)R(u_\lambda)-R(v_\lambda)I(u_\lambda)]\\
&=(v^{(1)}, u^{(1)})+i(v^{(1)}, u^{(2)}).
\end{align*}
We know that ${|(v,u)|}^2=(v^{(1)},u^{(1)})^2+(v^{(1)},u^{(2)})^2$. Then from $(v^{(2)},u^{(2)})=(v^{(1)},\\u^{(1)})$ and $(v^{(2)},u^{(1)})=-(v^{(1)},u^{(2)})$ we get $|(v^{(i)},u^{(j)})|\leq|(v,u)|$ for all $1\leq i, j\leq 2$. If $(v, u)=0$, then $(v^{(i)}, u^{(j)})=0$ for all
$1\leq i, j\leq 2$. The orthogonality $(v^{(1)}, v^{(2)})=0=(u^{(1)}, u^{(2)})$ is derived directly from the definition of $v^{(i)}$ and $u^{(i)}$. This completes the proof of Lemma $\ref{4.2}$.
\end{proof}

Now we prove Theorem \ref{thm4.1}. Let $n=N_{\mathbb{C}}(d, \gamma)$ and $\mathscr{B}=\{B_1,\ldots,B_n\}$ be a set of $n$ orthonormal bases in $\mathbb{C}^d$ with $\gamma_\mathbb{C}(\mathscr{B})\leq \gamma$. For each basis $B_\lambda=\{a_1,\ldots, a_d\}$ of $\mathbb{C}^d$, we get a set $B_\lambda(\mathbb{R})$ in $\mathbb{R}^{2d}$ where
$$B_\lambda(\mathbb{R})=\{a_\lambda^{(1)},  a_\lambda^{(2)}: 1\leq \lambda\leq d\}.$$

By Lemma \ref{4.2} and $B_\lambda$ is an orthonormal basis in $\mathbb{C}^d$, we know that $B_\lambda(\mathbb{R})$ is an orthonormal basis in $\mathbb{R}^{2d}$. Let $\mathscr{B}(\mathbb{R})=\{B_1(\mathbb{R}),\ldots,B_n(\mathbb{R})\}$. Then from Lemma \ref{4.2} (2) we get $\gamma_\mathbb{R}(\mathscr{B}(\mathbb{R}))\leq \gamma_\mathbb{C}(\mathscr{B})\leq\gamma$. Therefore $N_\mathbb{R}(2d; \gamma)\geq n=N_\mathbb{C}(d; \gamma)$. This completes the proof of Theorem \ref{thm4.1}.
\end{proof}

The following table shows a list of parameters of real AMUB $\mathscr{B}$ in $\mathbb{R}^D (D=2d)$ derived from known series of complex AMUB in $\mathbb{C}^d$
is given by Theorem \ref{bounds},\ref{th3.2} and \ref{th3.3}.

\begin{table}[htbp]
\setlength\tabcolsep{1.8pt}
\centering
\caption{real AMUB's $\mathscr{B}$ in $\mathbb{R}^D (D=2d)$ with $n=|\mathscr{B}|$ and $\gamma_{\mathbb{R}}(\mathscr{B})\leq \gamma$}
\begin{tabular}{ccccccc}%l=left, r=right,c=center分别代表左对齐，右对齐和居中，字母的个数代表列数
\hline
No.&$D=2d$ & $n$ & $\gamma$ &\ $\rm{derived \ from}$  \\ \hline
(1) &$2p^m$& $p^m+1$& $\frac{1}{\sqrt{p^m}}=\frac{\sqrt{2}}{\sqrt{D}}$& $\rm{Theorem \ \ref{bounds}} (2)$\\ \hline
(2) &$2d~(d\geq 2)$& 3& $\frac{\sqrt{2}}{\sqrt{D}}$& \rm{Theorem \ \ref{bounds}}(3)\\ \hline
(3) &\makecell[c]{$2d^2$\\$(d\geq 6, d\not\equiv2\pmod 4)$}& \makecell[c]{6}& \makecell[c]{$\frac{\sqrt{2}}{\sqrt{D}}$}& \makecell[c]{\rm{Theorem \ \ref{bounds}}(4)}  \\ \hline
(4) &\makecell[c]{$2d~(d\geq3)$\\$(|d-p|\leq2\sqrt{p})$}& \makecell[c]{$p^{m-1}$\\$(2\leq m\leq d-1)$}& \makecell[c]{$\frac{2\sqrt{2}m}{\sqrt{D}}+\frac{4m}{D}$}& \makecell[c]{\rm{Theorem \ \ref{th3.2}} }  \\ \hline
(5) &$2d~(d=p^m-1)$& $p^m$& ${(\frac{1}{d}+\frac{1+2\sqrt{d+1}}{d^2})}^{1/2}$& {\rm{Theorem}}~\ref{th3.3}(1) \\ \hline
(6) &$2d~(d=p^m-1)$& $p^m+1$& ${(\frac{1}{d}+\frac{1}{d^2})}^{1/2}$& {\rm{Theorem}}~\ref{th3.3}(2) \\ \hline
(7) &$2d~(d=p^m+1)$& $p^m$& ${(\frac{1}{d}+\frac{2\sqrt{d-1}}{d^2})}^{1/2}$& {\rm{Theorem}}~\ref{th3.3}(3)\\ \hline
(8) &$2d~(d=2^r(2^r-1))$& $2^r+1$& $\frac{1}{2^r-1}$& {\rm{Theorem}}~\ref{th3.3}(4) \\ \hline
\end{tabular}
\end{table}

Table $1$ presents many series of real AMUB's with nice parameters comparing to real MUB's and previous construction of real AMUB's.

(A). For $D\neq{(4m)^2}$, the size of real MUB $\mathscr{B}$ in $\mathbb{R}^D$ is less than $3$ and $\gamma_\mathbb{R}(\mathscr{B})=\frac{1}{\sqrt{D}}$.
In Table $1$ (1) and (5)-(7), we get many series of real AMUB's $\mathscr{B}$ in $\mathbb{R}^D$ with size $\frac{D}{2}$ or $\frac{D}{2}\pm 1$ and $\gamma_\mathbb{R}(\mathscr{B})\sim\frac{\sqrt{2}}{\sqrt{D}}$.

(B). For $d={(4p)^2}$ and $p$ is an odd prime, Corollary \ref{cor3.5} presents real AMUB $\mathscr{B}$ in $\mathbb{R}^d$ with size $|\mathscr{B}|=\frac{\sqrt{d}}{4}+1$ and $\gamma_{\mathbb{R}}(\mathscr{B})\leq\frac{4}{\sqrt{d}}$. On the other hand, Table $1$ (4) (for $m=2$) presents real
AMUB $\mathscr{B}$ in $\mathbb{R}^d$ (for even $d\geq 6$) with $\gamma_{\mathbb{R}}(\mathscr{B})\leq\frac{4\sqrt{2}}{\sqrt{d}}+\frac{8}{d}$ and larger
size $|\mathscr{B}|(=p)\geq\frac{d}{2}+2-2\sqrt{\frac{d}{2}+1}$. Table $1$ (8) presents real AMUB $\mathscr{B}$ in $\mathbb{R}^d$ $(d=2^{r+1}(2^r-1))$ with larger
size $|\mathscr{B}|=2^r+1\sim\sqrt{\frac{d}{2}}$ and smaller $\gamma_{\mathbb{R}}(\mathscr{B})=\frac{1}{2^r-1}\sim\frac{\sqrt{2}}{\sqrt{d}}.$
\section{Conclusion}
The MUB's $\mathscr{B}$ in $F^d$ $(F=\mathbb{C}$ or $\mathbb{R}$) with larger size $n=|\mathscr{B}|$, particularly for real case, are rare. Several series of approximate MUB's (AMUB's) in $F^d$ with $\gamma_F(\mathscr{B})=O(\frac{1}{\sqrt{d}})$ have been constructed [10,  12, 7, 6]. In this paper we present
a general construction of real AMUB $\mathscr{B}(\mathbb{R})$ in ${\mathbb{R}}^{2d}$ from any complex AMUB $\mathscr{B}$ in $\mathbb{C}^d$ with the same size
$|\mathscr{B}(\mathbb{R})=|\mathscr{B}|$ and $\gamma_\mathbb{R}(\mathscr{B}(\mathbb{R}))\leq \gamma_{\mathbb{C}}(\mathscr{B})$ (Theorem \ref{thm4.1}). Namely,
from any $\frac{c}{\sqrt{d}}$-AMUB in $\mathbb{C}^d$ we get a $\frac{c\sqrt{2}}{\sqrt{D}}$-AMUB in ${\mathbb{C}}^D (D=2d)$ with the same size. With this result, several series of real AMUB's in $\mathbb{R}^d$ with even dimension $D$ and larger size can be directly derived from known complex AMUB's in contrast with that
the size of any real MUB in $\mathbb{R}^d$  for $d\neq {(4m)}^2$ is at most $3$.

For odd integer $d$, to construct real AMUB in $\mathbb{R}^d$ with larger size seems to be difficult. We will deal with this problem in sequential papers.

\end{document}